\documentclass[journal]{IEEEtran}
\newfont{\bb}{msbm10 scaled 1100}

\usepackage{graphicx}
\usepackage{algorithm,algorithmic}
\usepackage{psfrag}
\usepackage{amsmath}
\usepackage{amsfonts,amssymb,amsmath}
\usepackage{mathrsfs}
\usepackage{amsthm}
\usepackage{array}
\usepackage{xcolor}

\newtheorem{proposition}{Proposition}
\newtheorem{theorem}{Theorem}
\newtheorem{lemma}{Lemma}

\theoremstyle{remark}
\newtheorem{remark}{Remark}
\newtheorem*{remark*}{Remark}

\begin{document}
\title{\huge{Convex Separable Problems with Linear and Box Constraints in Signal Processing and Communications}}
\author{Antonio A.~D'Amico, Luca~Sanguinetti \emph{Member, IEEE}, Daniel P. Palomar \emph{Fellow, IEEE}.
\thanks{\newline \indent A. A. D'Amico and L. Sanguinetti are with the
University of Pisa, Department of Information Engineering, Via
Caruso 56126 Pisa, Italy (e-mail: \{a.damico, luca.sanguinetti\}@iet.unipi.it).
\newline \indent L. Sanguinetti is also with the Alcatel-Lucent Chair on Flexible Radio, Sup\'elec, Gif-sur-Yvette, France (e-mail: luca.sanguinetti@supelec.fr.)
\newline \indent D. P. Palomar is with the Hong Kong University of Science and Technology, Department of Electronic and Computer Engineering, Clear Water Bay, Kowloon, Hong Kong (e-mail: palomar@ust.hk).
\newline \indent Part of the material of this paper was presented at the 39th IEEE International Conference on Acoustics, Speech, and Signal Processing (ICASSP), Florence, Italy, May 4 -- 9, 2014.%
}
\vspace{-0.6cm}}
\maketitle

\begin{abstract} In this work, we focus on separable convex optimization problems with box constraints and a set of triangular linear constraints. The solution is given in closed-form as a function of some Lagrange multipliers that can be computed through an iterative procedure in a finite number of steps. Graphical interpretations are given casting valuable insights into the proposed algorithm and allowing to retain some of the intuition spelled out by the water-filling policy. It turns out that it is not only general enough to compute the solution to different instances of the problem at hand but also remarkably simple in the way it operates. We also show how some power allocation problems in signal processing and communications can be solved with the proposed algorithm.
\end{abstract}

\section{Introduction}

\PARstart{C}{onsider} the following problem:

\begin{align}\label{1}
  (\mathcal P): \quad \underset{\{x_n\}}{\min} \quad &
\sum\limits_{n=1}^N f_n(x_n)\\\nonumber
\text{subject to} \quad & \sum \limits_{n=1}^{j} x_n\le \rho_j \quad j=1,2,\ldots,N\\\nonumber
 & l_n \le x_n \le u_n \quad n=1,2,\ldots,N
\end{align}
where $\{x_n\}$ are the optimization variables, the coefficients $\{\rho_j\}$ are real-valued parameters, and the constraints $ l_n \le x_n \le u_n$ are called variable bounds or box constraints with $-\infty \le l_n < u_n \le +\infty$. The functions $f_{n}$ are real-valued, continuous and strictly convex in $[l_n,u_{n}]$, and continuously differentiable in $(l_n,u_{n})$. If $f_{n}$ is not defined in $l_n$ and/or in $u_n$, then it is extended by continuity as $f_n(l_n)=\mathop {\lim }\nolimits_{x_n \rightarrow l_n^+} f_n(x_n)$ and $f_n(u_n)=\mathop {\lim }\nolimits_{x_n \rightarrow u_n^-} f_n(x_n)$. {Possible extensions of $(\mathcal P)$ will be discussed in Section II.C.} 


\subsection{Motivation and contributions}

Constrained optimization problems of the form \eqref{1} arise in connection with a wide range of power allocation problems in different applications and settings in signal processing and communications. For example, they arise in connection with the design of multiple-input multiple-output (MIMO) systems dealing with the minimization of the power consumption while meeting the quality-of-service (QoS) requirements over each data stream (see for example \cite{PalomarQoS2004,PalomarAug2005, Palomar2005aa, Jiang2006,Bergman2009} for point-to-point communications and \cite{Fu2011,Shaolei2010, Sanguinetti2012,An2013, Sanguinetti2013} for amplify-and-forward relay networks). A survey of some of these problems for point-to-point MIMO communications can be found in \cite{Palomar2005}. It also appears in the design of optimal training sequences for channel estimation in multi-hop transmissions using decode-and-forward protocols \cite{Gao2008} and in the optimal power allocation for the maximization of the instantaneous received signal-to-noise ratio in amplify-and-forward multi-hop transmissions under short-term power constraints \cite{Farhadi09}.
Other instances of \eqref{1} are shown to be the rate-constrained power minimization problem over a code division multiple-access channel with correlated noise \cite{Padakandla2009} and the power allocation problem in amplify-and-forward relaying scheme for multiuser cooperative networks under frequency-selective block-fading \cite{Pham2010}. {Formulations as in \eqref{1} arise also in wireless communications with energy harvesting constraints. For example, they appear in \cite{Ozel2011} wherein the authors look for the optimal energy management scheme that maximizes the throughput in a point-to-point link with an energy harvesting transmitter operating over a fading channel. They can also be found in the design of the precoding strategy that maximizes the mutual information along independent channel accesses under non-causal knowledge of the channel state and harvested energy \cite{Gregori2013}.}

{Clearly, the optimization problem in \eqref{1} can always be solved using standard convex solvers. Although possible, this in general does not provide any insights into its solution and does not exploit the particular structure of the problem itself. In this respect, all the aforementioned works go a step further and provide ad-hoc algorithms for specific instances of \eqref{1} in the attempt of giving some intuition on the solutions. However, this is achieved at the price of a loss of generality in the sense that most of them can only be used for the specific problem at hand.} On the contrary, the main contribution of this work is to develop a general framework that allows one to compute the solution (and its structure) for any problem in the form of \eqref{1}. In other words, whenever a problem can be put in the form of \eqref{1}, then its solution can be efficiently obtained by particularizing the proposed algorithm to the problem at hand without the need of developing specific solutions. 

\vspace{-0.2cm}
\subsection{Related work}
The main related literature to this paper is represented by \cite{Padakandla2007} and \cite{Wang2012} in which the authors focus on solving problems of the form:
\begin{align}\label{1.10}
 \underset{\{x_n\}}{\min} \quad &
\sum\limits_{n=1}^N f_n(x_n)\\\nonumber
\text{subject to} \quad & \sum \limits_{n=1}^{j} x_n\le \sum\limits_{n=1}^j\alpha_n \quad j=1,2,\ldots,N\\\nonumber
 & 0 \le x_n \le u_n \quad n=1,2,\ldots,N
\end{align}
with $\alpha_n \ge 0$ for any $n$. The above problems are known as separable convex optimization problems with linear \emph{ascending inequality} constraints and box constraints. In particular, in \cite{Padakandla2007} the authors propose a dual method to numerically evaluate the solution of the above problem in no more than $N-1$ iterations under an ordering condition on the slopes of the functions at the origin. An alternative solution improving the worst case complexity of \cite{Padakandla2007} is illustrated in \cite{Wang2012}. {Differently from \cite{Padakandla2007} and \cite{Wang2012}, we consider more general problems in which the inequality constraints are \emph{not necessarily} in ascending order since the box constraint values $l_n$ and $u_n$ may possibly be equal to $-\infty$ and $+\infty$, respectively. All this makes \eqref{1} more general than problems of the form given in \eqref{1.10}. Observe, however, that if the lower bounds $l_n$ are all finite, then problem \eqref{1} boils down to \eqref{1.10} (as it can be easily shown using simple mathematical arguments). Compared to \cite{Padakandla2007} and \cite{Wang2012}, however, we also follow a different approach that allows us (simply exploiting the inherent structure of \eqref{1}) to focus only on functions $f_n$ that are continuous, strictly convex and monotonically decreasing in the intervals $[l_n, u_n]$. Furthermore, differently from \cite{Padakandla2007} we do not impose any constraints on the slopes of $f_n$.}

It is also worth mentioning that at the time of submission we became aware (through a private correspondence with the authors) of \cite{NCC2014} in which the problem originally solved in \cite{Padakandla2007} has been revisited in light of the theory of polymatroids. In particular, in \cite{NCC2014} the authors have removed some of the restrictions on functions $f_n$ that were present in \cite{Padakandla2007}. This allows them to come up with a solution similar to the one we propose in this work.

\subsection{Organization}
The remainder of the paper is structured as follows. Some preliminary results are discussed in the next section {together with some possible extensions of the problem at hand}. Section \ref{main_result} provides the main result of the paper: an algorithm to evaluate the solution to ($\mathcal{P}$). Section \ref{graphical_interpretation} presents some graphical interpretations of the way the proposed solution operates. This leads to an interesting water-filling inspired policy. {Section \ref{examples} shows how some power allocation problems of practical interest in signal processing and communications can be solved with the proposed algorithm.} Finally, some conclusions are drawn in Section \ref{conclusions}.


\section{Preliminary results and discussions}
Some preliminary results are discussed in the sequel. In particular, we first study the feasibility (admissibility) of \eqref{1} and then we show that the optimization in \eqref{1} reduces to solve an equivalent problem in which all the functions $f_n$ are continuous, strictly convex and monotonically decreasing in the intervals $[l_n,u_n]$. {In addition, we also discuss some possible extensions of \eqref{1}.}
\vspace{-0.3cm}
\subsection{Feasibility }
The feasibility of \eqref{1} simply amounts to verifying that for given values of $\{l_n\}$, $\{u_n\}$ and $\{\rho_n\}$, the \textit{feasible set} (or \textit{constraint set}) is not empty \cite{BoydBook}.
A necessary and sufficient condition for \eqref{1} to be feasible is provided in the following proposition.

\begin{proposition}
The solution to \eqref{1} exists if and only if
\begin{align}\label{feasibility}
\sum \limits_{n=1}^{j} l_n\le \rho_j \quad j=1,2,\ldots,N.
\end{align}
\end{proposition}

\begin{proof}
The proof easily follows from $(\mathcal{P})$ {since the point $(l_1,l_2,\ldots,l_N)$ is feasible.}
\end{proof}

In all subsequent discussions, we assume that \eqref{feasibility} is satisfied and denote $x^{\star}_n$, for $n=1,2,\ldots,N$, the solutions of \eqref{1}.

\vspace{-0.2cm}
\subsection{Monotonic properties of $f_n$}
Observe that since $f_n$ is by definition strictly convex in $[l_n,u_{n}]$ and continuously differentiable in $(l_n,u_{n})$, then the three following cases may occur.

  \vspace{0.1cm}
{\bf{a}}) The function $f_n$ is monotonically increasing in $[l_n,u_{n}]$ or, equivalently, $f_n'(x_n)>0$ for any $x_n \in (l_n,u_{n})$.

  \vspace{0.1cm}
{\bf{b}}) There exists a point $z_n$ in $(l_n,u_n)$ such that $f_n^{'} (z_n) = 0 $ with $f_n'(x_n)<0$ and $f_n'(x_n)>0$ for any $x_n$ in $(l_n, z_n)$ and $(z_n, u_n)$, respectively.

  \vspace{0.1cm}
{\bf{c}}) The function $f_n$ is monotonically decreasing in $[l_n,u_{n}]$ or, equivalently, $f_n'(x_n)<0$ for any $x_n \in (l_n,u_{n})$.

  \begin{lemma} \label{LemmaA}
  If $f_n$ is monotonically increasing in $[l_n,u_{n}]$ and $l_n \ne - \infty$, then $x_n^\star$ is given by
  \begin{align}\label{solA}
   x_n^\star = l_n.
  \end{align}
  \end{lemma}
  \vspace{-0.15cm}
\begin{proof} The proof is given in Appendix A. \end{proof}

The above result can be used to find an equivalent form of \eqref{1}. Denote by $\mathcal A \subseteqq \mathcal \{1,2,\ldots,N\}$ the set of indices $n$ in \eqref{1} for which case {\bf{a}}) holds true and assume (without loss of generality)
that $\mathcal A=\{1,2,\ldots,|\mathcal A|\}$.
Using the results of Lemma 1, it follows that $x_n^\star = l_n$ for any $n \in \mathcal A$ while the computation of the remaining variables with indices $n \notin \mathcal A$ requires to solve the following reduced problem:
  \begin{align}\label{caseA}
   \underset{\{x_n\}}{\min} \quad &
                   \sum\limits_{n=|\mathcal A|+1}^{N} f_n(x_n)\\\nonumber
   \text{subject to} \quad & \sum \limits_{n=|\mathcal A|+1}^{j} x_n\le \rho_j^{\prime} \quad j=|\mathcal A|+1,\ldots,N\\\nonumber
   & l_n \le x_n \le u_n \quad n=|\mathcal A|+1,\ldots,N
  \end{align}
  with
    \begin{align}
\rho_j^{\prime} = \rho_j-\sum\limits_{n=1}^{|\mathcal A|}l_n
  \end{align}
for $ j = |\mathcal A|+1,\ldots,N$\footnote{Notice that in order for problem in \eqref{caseA} and thus for the original problem in \eqref{1} to be well-defined it must be $l_n \ne - \infty$ $\forall n \in \mathcal A$.}. The above optimization problem is exactly in the same form of \eqref{1} except for the fact that all its functions $f_n$ fall into cases {\bf{b}}) or {\bf{c}}). To proceed further, we make use of the following result.

  \begin{lemma} \label{LemmaB}
  If there exists a point $z_n$ in $(l_n,u_n)$ such that $f_n^{'} (z_n) = 0 $ with $f_n'(x_n)<0\; \forall x_n \in (l_n, z_n)$ and $f_n'(x_n)>0\; \forall x_n \in (z_n, u_n)$, then it is always
  \begin{align}\label{solB}
   l_n \le x_n^\star \le z_n.
  \end{align}
  \end{lemma}
  
\begin{proof} 
The proof is given in Appendix A. 
\end{proof}

Using the above result, it follows that solving \eqref{caseA} amounts to looking for the solution of the following equivalent problem:
{\begin{align}\label{caseB}
 \underset{\{x_n\}}{\min} \quad &
                   \sum\limits_{n=|\mathcal A|+1}^{N} f_n(x_n)\\\nonumber
   \text{subject to} \quad & \sum \limits_{n=|\mathcal A|+1}^{j} x_n\le \rho_j^{\prime} \quad j=|\mathcal A|+1,\ldots,N\\\nonumber
   & l_n \le x_n \le u^\prime_n \quad n=|\mathcal A|+1,\ldots,N
\end{align}}
where
\begin{align}\label{}
 u^\prime_n &= z_n \quad  \text{if} \; n\; \in \; \mathcal{B} \\
 u^\prime_n & = u_n \quad \text{otherwise}
  \end{align}
with $\mathcal B$ being the set of indices $n$ in \eqref{caseA} for which case {\bf{b}}) holds true. The above problem is in the same form as \eqref{1} with the only difference that all functions $f_n$ are monotonically decreasing in $(l_n,u'_{n})$ and thus fall into case {\bf{c}}).

The results of Lemmas \ref{LemmaA} and \ref{LemmaB} can be summarized as follows. Once
the optimal values of the variables associated with functions $f_n$ that are monotonically increasing have been trivially computed through \eqref{solA}, it remains to solve the optimization problem \eqref{caseA} in which the functions $f_n$ belong to either case {\bf{b}}) or {\bf{c}}). In turn, problem \eqref{caseA} is equivalent to problem \eqref{caseB} with only class {\bf{c}}) functions. This means that we can {simply} consider optimization problems of the form in \eqref{1} in which all functions $f_n$ fall into case {\bf{c}}). Accordingly, in the following we assume that \eqref{feasibility} is satisfied and only focus on functions $f_n$ that are continuous, strictly convex and monotonically decreasing in the intervals $[l_n,u_n]$. {For notational simplicity, however, in all subsequent derivations we maintain the notation given in \eqref{1}, though we assume that the results of Lemmas 1 and 2 have been already applied.}

\vspace{-0.4cm}
\subsection{Possible extensions}

{An equivalent form of $(\mathcal P)$, which is sometimes encountered in literature, is given by:
\begin{align}\label{P_2}
 \underset{\{x_n\}}{\min} \quad &
\sum\limits_{n=1}^N f_n(x_n)\\\nonumber
\text{subject to} \quad & \sum \limits_{n=1}^{j} x_n\ge \rho_j \quad j=1,2,\ldots,N\\\nonumber
 & l_n \le x_n \le u_n \quad n=1,2,\ldots,N. \nonumber
\end{align}
{The above problem can be rewritten in the same form as in \eqref{1} simply replacing $x_n$ with $y_n = - x_n$ in \eqref{P_2}. In doing this, we obtain
\begin{align} \label{P_3}
 \underset{\{y_n\}}{\min} \quad &
\sum\limits_{n=1}^N f_n(- y_n)\\ \nonumber
\text{subject to} \quad & \sum \limits_{n=1}^{j} y_n\le - \rho_j \quad j=1,2,\ldots,N\\ \nonumber
 & -u_n \le y_n \le -l_n\quad n=1,2,\ldots,N
\end{align}
which is exactly in the same form of $(\mathcal P)$.}} 

{Consider also the following problem
\begin{align}\label{P_3}
 \underset{\{x_n\}}{\min} \quad &
\sum\limits_{n=1}^N f_n(x_n)\\\nonumber
\text{subject to} \quad & \sum \limits_{n=1}^{j} g_n(x_n)\le \rho_j \quad j=1,2,\ldots,N\\\nonumber
 & l_n \le x_n \le u_n \quad n=1,2,\ldots,N \nonumber
\end{align}
in which $g_n$ is a {continuos and} strictly increasing function. Setting $y_n=g_n(x_n)$ yields
\begin{align}\label{P_3.1}
 \underset{\{x_n\}}{\min} \quad &
\sum\limits_{n=1}^N p_n(y_n)\\\nonumber
\text{subject to} \quad & \sum \limits_{n=1}^{j} y_n\le \rho_j \quad j=1,2,\ldots,N\\\nonumber
 & l_n^\prime \le {y_n} \le u_n^\prime \quad n=1,2,\ldots,N. \nonumber
\end{align}
where $p_n = f_n \circ g_n^{-1}$, $l_n^\prime = g_n(l_n)$ and $u_n^\prime = g_n(u_n)$ with $l_n^\prime < u_n^\prime$ since $g_n$ is strictly increasing. Clearly, \eqref{P_3.1} is in the same form of $(\mathcal P)$ in \eqref{1} provided that $p_n$
is continuous and strictly convex in $[l_n^\prime ,u^\prime_n]$, and continuously differentiable in ($l_n^\prime ,u^\prime_n$). This happens for example when: $i$) $f_n$ is a strictly convex decreasing function and
$g_n^{-1}$ is a concave function (or, equivalently, $g_n$ is a convex function); $ii$) $f_n$ is a strictly convex increasing function and $g_n^{-1}$ is a convex function (or, equivalently, $g_n$ is a concave function).}

{Similar arguments can be used when $g_n$ in \eqref{P_3} is a strictly decreasing function. This means that the results of this work can also be applied to the case in which the constraints have the following form:
\begin{align}\label{constraint_g_n}
\sum\limits_{n=1}^j g_n(x_n) \le \rho_j
\end{align}
with $g_n$ being continuously differentiable and invertible in $[l_n ,u_n]$.}

\vspace{-0.2cm}
\section{The main result}\label{main_result}

This section proposes an iterative algorithm that computes the solutions $x_n^\star$ for $n=1,2,\ldots,N$ in a finite number of steps $L<N$.
We begin by {denoting}
\begin{align}\label{h_n}
h_n(x_n) = -f_n^{'}(x_n)
\end{align}
which is a positive and strictly decreasing function since $f_n$ is by definition monotonically decreasing, strictly convex in $[l_n,u_{n}]$ and continuously differentiable in $(l_n,u_{n})$. {We take $h_n(l_n)=\mathop {\lim }\nolimits_{x_n \rightarrow l_n^+} h_n(x_n)$ and $h_n(u_n)=\mathop {\lim }\nolimits_{x_n \rightarrow u_n^-} h_n(x_n)$.} We also define the functions $\xi_n(\varsigma)$ for $n=1,2,\ldots,N$ as follows
\begin{align}\label{xi_n}
\xi_n(\varsigma) = \left\{ {\begin{array}{*{20}c}
  {u_n } & {0 \le \varsigma < h_n (u_n )} \\ \\
  {h_n^{-1}(\varsigma)} & { h_n (u_n ) \le \varsigma < h_n(l_n )}\\ \\
   {l_n} & {h_n(l_n ) \le \varsigma} \\
\end{array}} \right.
\end{align}
where $0 \le \varsigma < +\infty$ and ${h_n^{-1}} $ denotes the inverse function of ${h_n}$ within the interval $[l_n,u_n]$. Since $h_n$ is a continuous and strictly decreasing function, then ${h_n^{-1}} $ is continuous and strictly decreasing {whereas $\xi_n$ is continuous and non-increasing}. Functions $\xi_n(\varsigma) $ in \eqref{xi_n} can be easily rewritten
in the following compact form:
\begin{align}\label{xi_n.1}
\xi_n(\varsigma) = \min\left\{\max\left\{h_n^{-1}(\varsigma), l_n\right\}, u_n\right\}
\end{align}
from which it is seen that each $\xi_n(\varsigma)$ projects $h_n^{-1}(\varsigma)$ onto the interval $[l_n,u_n]$.

\begin{theorem}
The solutions of $(\mathcal P)$ are given by
\begin{align}\label{x_n^star}
x_n^\star= \xi_n(\sigma_n^\star) 
\end{align}
where the quantities {$\sigma_n^\star \ge 0$} for $n=1,2,\ldots,N$ are some Lagrange multipliers satisfying the following conditions{\footnote{We use $0\le x\, \bot \, y \le0$ to denote $0\le x$, $y \le 0$ and $xy = 0$.}}:
{\begin{equation}\label{KKT2eq_1}
 0 \le (\sigma_n^\star-\sigma_{n+1}^\star) \; \bot\; \Big(\sum\limits_{j=1}^{n}x_j^\star-\rho_n\Big) \le 0 \quad
 \end{equation}}
with $\sigma_{N+1}^\star = 0$.
\end{theorem}

\begin{proof} The proof is given in Appendix B. \end{proof}

 From \eqref{xi_n.1} and \eqref{x_n^star}, it easily follows that $x_n^\star$ can be compactly represented as
\begin{equation}\label{x_n^star_2}
x_n^\star= \min\left\{\max\left\{h_n^{-1}(\sigma_n^\star), l_n\right\}, u_n\right\}.
\end{equation}


\begin{lemma}
The Lagrange multipliers $\sigma_n^\star$ satisfying \eqref{KKT2eq_1} can be computed by means of the iterative procedure illustrated in {\bf{Algorithm 1}}.\end{lemma}

\begin{proof} The proof is given in Appendix C. \end{proof}

{\begin{algorithm}[t]
\caption{Iterative procedure for solving $(\mathcal P)$ in \eqref{1}.}

\begin{enumerate}
\item Set $j=0$ and $\gamma_n =\rho_n$ for every $n$.
\item {{\bf{While}}} $j < N$
\begin{enumerate}
\item Set $\mathcal{N}_{j}=\{j+1, \ldots, N\}$
\item For every $n \in \mathcal{N}_{j}$.
\begin{enumerate}
\item If $\gamma_n < \sum\nolimits_{i=j+1}^n u_i$ then compute $\varsigma_n^\star$ as the solution of
\begin{align}\label{100.10}
c_n(\varsigma)=\sum\limits_{i=j+1}^n \xi_i(\varsigma) = \gamma_n\end{align}
for $\varsigma$.
\item If $\gamma_n \ge \sum\nolimits_{i=j+1}^n u_i$ then set
\begin{align}\label{100.101}
\varsigma_n^\star=0.
\end{align}
\end{enumerate}
\item Evaluate
\begin{align}\label{101.10}
\mu^\star= \underset{n\in \mathcal{N}_{j}}{\max} \quad \varsigma_n^\star
\end{align}
and
\begin{align}\label{102.10}
k^\star = \underset{n\in \mathcal{N}_{j}}{ \max} \left\{n | \varsigma_n^\star=\mu^\star\right\}.
\end{align}
\item Set $\sigma_n^\star\leftarrow \mu^\star$ for $n=j+1, \ldots,$ $ k^\star$.\\
\item Use $\sigma_n^\star$ in \eqref{x_n^star} to obtain $x_n^\star$ for $n=j+1, \ldots,$ $ k^\star$. \\
\item Set $\gamma_n\leftarrow\gamma_n- \gamma_{k^\star}$
for $n=k^\star+1, \ldots, N$.\\
\item Set $j\leftarrow k^\star$.
\end{enumerate}
\end{enumerate}
\end{algorithm}}

As seen, Algorithm 1 proceeds as follows {(see also Section IV for a more intuitive graphical illustration)}. At the first iteration it sets $j=0$ and $\gamma_n=\rho_n$, $\forall n$, and for those values of $n \in \{1,2,\ldots,N\}$ such that
\begin{align}\label{c_n_1_10}
\gamma_n < \sum\limits_{i=1}^n u_i
\end{align}
it computes the unique solution $\varsigma_n^\star$ (see Appendix D for a detailed proof on the existence and uniqueness of $\varsigma_n^\star$) of the following equation
\begin{align}\label{c_n_1_2}
c_n(\varsigma) = \sum\limits_{i=1}^n \xi_i(\varsigma) = \gamma_n.
\end{align}
On the other hand, for those values of $n\in \{1,2,\ldots,N\}$ such that
\begin{align}\label{c_n_1_11}
\gamma_n \ge \sum\limits_{i=1}^n u_i
\end{align}
it sets $\varsigma_n^\star=0$. The values $\varsigma_n^\star$ computed as described above, for $n=1,2,\ldots,N$, are used in \eqref{101.10} and \eqref{102.10} to obtain $\mu^\star$ and $k^\star$, respectively. {As it follows from \eqref{101.10} and \eqref{102.10}, $\mu^\star$ is set equal to the maximum value of $\{\varsigma_n^\star\}$ with $n\in \{1,2,\ldots,N\}$ while $k^\star$ stands for its corresponding index.} Both are then used to replace $\sigma_n^\star$ with $\mu^\star$ for $n=1,2,\ldots,k^\star$. Note that if two or more indices can be associated with $\mu^\star$ (meaning that $\varsigma_n^\star = \mu^\star$ for all such indices), then according to \eqref{102.10} the maximum one is selected.

Once $\{\sigma_1^\star,\sigma_2^\star,\ldots,\sigma_{k^{\star}}^\star\}$ have been computed, Algorithm 1 moves to the second step, which essentially consists in solving the following reduced problem:
\begin{align}\label{4.1}
   \quad \underset{\{x_n\}}{\min} \quad &
\sum\limits_{n=k^\star +1}^N f_n(x_n)\\ \nonumber
\text{subject to} \quad & \!\!\!\!\!\!\sum \limits_{n=k^\star +1}^{j} x_n\le \gamma_j - \gamma_{k^{\star}}\quad j=k^\star +1,k^\star +2,\ldots,N\\ \nonumber
 & l_n \le x_n \le u_n \quad n=k^\star +1,k^\star+2,\ldots,N
\end{align}
using the same procedure as before. The iterative procedure terminates in a finite number of steps when all quantities $\sigma_n^\star$ are computed. According to Theorem 1, the solutions of $(\mathcal P)$ for $n=1,2,\ldots, N$ are eventually obtained as $x_n^\star= \xi_n(\sigma_n^\star)$.

\vspace{-0.2cm}
\subsection{Remarks}

The following remarks are of interest.

\begin{remark}It is worth observing than in deriving Algorithm 1 we have implicitly assumed that the number of linear constraints in \eqref{1} is exactly $N$. When this does not hold true, Algorithm 1 can be slightly modified in an intuitive and straightforward manner. Specifically, let $\mathcal{L} \subset \{1,2,\ldots,N\}$ denote the subset of indices associated to the linear constraints of the optimization problem at hand. In these circumstances, we have that \eqref{1} reduces to:
  \begin{align}\label{1mod}
   \underset{\{x_n\}}{\min} \quad &
   \sum\limits_{n=1}^N f_n(x_n)\\\nonumber
   \text{subject to} \quad & \sum \limits_{n=1}^{j} x_n\le \rho_j \quad j \in \mathcal{L}\\
   & l_n \le x_n \le u_n \quad n=1,2,\ldots,N. \nonumber
   \end{align}
The solution of \eqref{1mod} can still be computed through the iterative procedure illustrated in Algorithm 1 once the two following changes are made:

  \begin{itemize}
  \item Step a) -- Replace $\mathcal{N}_{j}$ with $\mathcal{N}_{j} \cap \mathcal{L}$.
  \item Step f) -- Replace the statement ``Set $\gamma_n\leftarrow\gamma_n- \gamma_{k^\star}$
   for $n=k^\star+1,k^\star+2, \ldots, N$" with ``Set $\gamma_n\leftarrow\gamma_n- \gamma_{k^\star}$
   for $n \in \{k^\star+1,k^\star+2, \ldots, N\} \cap \mathcal{L}$".
  \end{itemize}
As seen, when only a subset $\mathcal{L}$ of constraints must be satisfied, then Algorithm 1 proceeds computing the quantities $\varsigma_n^\star$ only for the indices $n\in \mathcal{L}$.
\end{remark}

\begin{remark}The number of iterations $L$ required by Algorithm 1 to compute all the Lagrange multipliers $\sigma_n^\star$ (and, hence, to compute all the solutions $x_n^\star$) depends on the cardinality $\left|\mathcal{L}\right|$ of $\mathcal{L}$ (or, equivalently, on the number of linear constraints). {In general, $L$ is less than or equal to $\left|\mathcal{L}\right|$. However,} if $\left|\mathcal{L}\right|=1$ only one iteration is required and thus $L=1$. Also, if there is no linear constraint (which means $\mathcal{L}= \emptyset $ and $\left|\mathcal{L}\right|=0$) the solutions of $(\mathcal{P})$ can be computed without running Algorithm 1 since they are trivially given by $x_n^\star=u_n$. On the other hand, if $\left|\mathcal{L}\right|=N$ the maximum number of iterations required is $N-1$. Indeed, assume that at each iteration Algorithm 1 provides only one $x_n^\star$ (which amounts to saying that at the first iteration Algorithm 1 computes $x_1^\star$, at the second $x_2^\star$, and so forth). Accordingly, at the end of the $(N-1)$th iteration the values of $x_1^\star,x_2^\star,\ldots,x_{N-1}^\star$ are available, and $x_N^\star$ can be directly computed as $ x_N^\star = \min\{(\rho_N-\sum_{n=1}^{N-1}x_n^\star),u_n\}$ {without the need of performing the $N$th iteration. For simplicity, in the sequel we assume that the $N$th iteration is always performed so that it is assured that the last value of $k^\star$ computed through \eqref{102.10} is always equal to $N$.}
\end{remark}
\begin{remark} Observe that if there exists one or more values of $j\in \mathcal{L}$ in \eqref{1mod} for which the following condition holds true 
\begin{align}\label{000012}
\rho_j= \sum\limits_{i=1}^j l_i
\end{align}
then it easily follows that $x^{\star}_{n}=l_n$ for $n=1,2,\ldots, j_{\max}$, with $j_{\max}$ being the maximum value of $j\in \mathcal{L}$ such that the above condition is satisfied. This means that solving \eqref{1mod} basically reduces to find the solution of the following problem:
  \begin{align}\label{caseCAC}
   \underset{\{x_n\}}{\min} \quad &
                   \sum\limits_{n=j_{\max}+1}^{N} f_n(x_n)\\\nonumber
   \text{subject to} \quad & \sum \limits_{n=j_{\max}+1}^{j} x_n\le \rho_j^{\prime} \quad j \in \mathcal L \setminus \mathcal C\\\nonumber
   & l_n \le x_n \le u_n \quad n=j_{\max}+1,\ldots,N
  \end{align}
in which $\mathcal C = \{1,2, \ldots, j_{\max}\}$ and $\rho_j^{\prime} = \rho_j-\sum\nolimits_{n=1}^{j_{\max}}l_n$.
\end{remark}

\begin{remark} {For later convenience, we concentrate on the computation of $\mu^\star$ in the last step of Algorithm 1}. For this purpose, denote by $\{\mu^{\star}_{1},\mu^{\star}_{2},\ldots,\mu^{\star}_{L}\}$ and $ k^{\star}_{1}<k^{\star}_{2}<\cdots<k^{\star}_{L}$ (with $k^{\star}_{1} \ge 1$ and $k^{\star}_{L}=N$) the values of $\mu^{\star}$ and $k^\star$ provided by \eqref{101.10} and \eqref{102.10}, respectively, at the end of the $L$ iterations required to solve $(\mathcal{P})$. Setting $k_{0}^{\star}=0$, we may write
\begin{equation}
\label{rem3_1}
\sigma_{n}^\star=\mu_{j}^{\star} \quad k_{j-1}^{\star}+1 \le n \le k_{j}^{\star} \; \mathrm{and} \; j=1,2,\ldots,L
\end{equation}
with $\{\mu_{1}^{\star},\mu_{2}^{\star},\ldots,\mu_{{L-1}}^{\star}\}$ such that
\begin{align}\label{rem3_2}
\sum\limits_{n= k_{j-1}^{\star}+1}^{k_{j}^{\star}} x_n^\star = \rho_{k_{j}^{\star}}- \rho_{k_{j-1}^{\star}} \quad j=1,2,\ldots,{L-1}.
\end{align}
For $j=L$ two cases may occur, namely $\mu_{L}^{\star} > 0$ or $\mu_{L}^{\star} = 0$. In the former, $\mu_{L}^{\star} $ is such that
\begin{align} \label{rem3_3}
\sum\limits_{n= k_{L-1}^{\star}+1}^{k_{L}^{\star}} x_n^\star = \rho_{k_{L}^{\star}}- \rho_{k_{L-1}^{\star}}
\end{align}
while in the latter we simply have that
\begin{equation}
\label{rem3_4}
x_n^\star=u_n \quad \quad n=k_{L-1}^{\star}+1,\ldots,k_{L}^{\star}.
\end{equation}
\end{remark}
\begin{remark}
At any given iteration, Algorithm 1 requires to solve at most $N-k^\star$ non-linear equations (where $k^\star$ is the value obtained from \eqref{102.10} at the previous iteration):
\begin{align}\label{eq1R1}
\!\!\!\!\!c_n({\varsigma})= \sum\limits_{i=k^\star+1}^n \xi_i(\varsigma) = \gamma_n\quad n=k^\star+1,k^\star+2,\ldots, N.
\end{align}
When the solutions $\{\varsigma_{n}^\star\}$ of the above equations can be computed in closed form, the computational complexity required by each iteration is nearly negligible. On the other hand, when a closed-form does not exist, this may result in excessive computation. In this latter case, a possible means of reducing the computational complexity relies on the fact that $c_n(\varsigma)$ is a non-increasing function as it is the sum of non-increasing functions. Now, assume that the solution of \eqref{eq1R1} has been computed for $n=n'$. Since we are interested in the maximum between the solutions of \eqref{eq1R1}, as indicated in \eqref{101.10}, then for $n''>n'$ $c_{n''}(\varsigma) = \gamma_{n''}$ must be solved only if $c_{n''}(\varsigma_{n'}^\star) > \gamma_{n''}$. Indeed, only in this case $\varsigma_{n''}^\star$ would be greater than $\varsigma_{n'}^\star$. Accordingly, we may proceed as follows. We start by solving \eqref{eq1R1} for $n=k^\star+1$. Then, we look for the first index $n>k^\star+1$ for which $c_{n}(\varsigma_{k^\star+1}^\star) > \gamma_{n}$ and solve the equation associated to such an index. We proceed in this way until $n=N$. In this way, the number of non-linear equations solved at each iteration is smaller than or equal to that required by Algorithm 1.
\end{remark}

\begin{remark}
{From the above remark, it follows that the proposed algorithm can be basically seen as composed of two layers. The outer layer computes the Lagrange multipliers $\{\sigma_n^\star\}$ whereas the inner layer evaluates the solution to \eqref{eq1R1}. If the latter can be solved in closed form, then the complexity required by the inner layer is negligible and thus the number of iterations required to solve the problem is essentially given by the number of iterations of the outer layer, which is at most $N-1$  with $N$ being the number of linear constraints. On the other hand, if the solution to \eqref{eq1R1} cannot be computed in closed form, then the total number of iterations should also take into account the complexity of the inner layer. However, this cannot be easily quantified as it largely depends on the particular structure of \eqref{eq1R1} and the specific iterative procedure used to solve it.} 
\end{remark}

%

\begin{figure}[t]
\begin{center}
\psfrag{x1}[r][m]{\scriptsize{$x_1^\star = -0.8$}}
\psfrag{x2}[r][m]{\scriptsize{$x_2^\star=-1.2$}}
\psfrag{x3}[r][m]{\scriptsize{$x_3^\star=1.9$}}
\psfrag{x4}[r][m]{\scriptsize{$x_4^\star=-1.8$}}
\psfrag{x5}[c][m][1][90]{\scriptsize{\quad\;$\sigma_3^\star=\sigma_4^\star= 1.195$\quad}}
\psfrag{x6}[c][m][1][90]{\scriptsize{$\sigma_1^\star=\sigma_2^\star= 4.451$}}
\psfrag{x7}[r][m]{\scriptsize{$\varsigma$}}
\psfrag{data1}[l][m]{\scriptsize{$\xi_1(\varsigma)$}}
\psfrag{data2}[l][m]{\scriptsize{$\xi_2(\varsigma)$}}
\psfrag{data3}[l][m]{\scriptsize{$\xi_3(\varsigma)$}}
\psfrag{data4}[l][m]{\scriptsize{$\xi_4(\varsigma)$}}
\includegraphics[width=0.9\columnwidth]{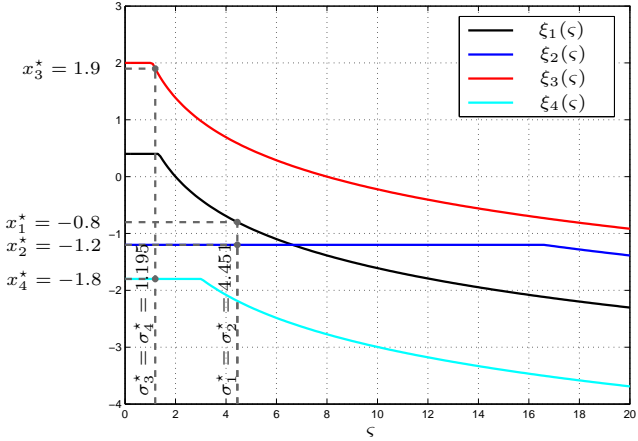}
\end{center}
\caption{Graphical illustration of the solutions $x_n^\star$. The intersection of $\xi_n(\varsigma)$ with the vertical dashed line at $\varsigma= \sigma_n^\star$ yields $x_n^\star$.}
\label{fig1}
\end{figure}

\vspace{-0.1cm}
\section{Graphical interpretations}\label{graphical_interpretation}

In the next, we provide graphical interpretations of the general policy spelled out by Theorem 1 and Lemma 3. 

\subsection{Charts}
A direct depiction of Theorem 1 and Lemma 3 can be easily obtained by plotting $c_n(\varsigma)$ and $\xi_n(\varsigma)$ for $n=1,2,\ldots,N$ as a function of $\varsigma \ge 0$. From \eqref{100.10}, it follows that the intersections of curves $c_n(\varsigma)$ with the horizontal lines at $\gamma_n$ yield $\varsigma_n^\star$ from which $\mu^\star$ and $k^\star$ are computed as indicated in \eqref{101.10} and \eqref{102.10}. According to \eqref{x_n^star}, the solutions $x_n^\star$ for $n=1,2,\ldots,k^\star$ correspond to the interception of the corresponding functions $\xi_n(\varsigma)$ with the vertical line at $\varsigma = \sigma_n^\star = \mu^\star$. Once $x_n^\star$ for $n=1,2,\ldots,k^\star$ are computed, the algorithm proceeds with the computation of the remaining solutions by solving the corresponding reduced problem.

For illustration purposes, we assume $N=4$, $l_n = -\infty$ for any $n$, $\mathbf{u}=[0.4, -1.2, 2, -1.8]$ and ${\boldsymbol{\rho}} = [0.2, -2, 1.1, -1.9]$. In addition, we set 
\begin{align}\label{5.1}
f_n(x_n) = w_n e^{-x_n}
\end{align}
with $[w_1, w_2, w_3, w_4]=[2, 5, 8, 0.5]$. Then, it follows that $h_n(x_n) ={w_n}e^{-x_n}$ and $h_n^{-1}(\varsigma) = \ln w_n -\ln \varsigma$.
Then, from \eqref{xi_n} we obtain
\begin{align}\label{xi_n.0}
\xi_n(\varsigma) = \left\{ {\begin{array}{*{20}c}
  {u_n } & {0 \le \varsigma < w_n e^{-u_n} } \\ \\
  \ln w_n -\ln \varsigma & { w_n e^{-u_n} \le \varsigma}\\
\end{array}} \right.
\end{align}
or, more compactly,
\begin{align}\label{xi_n.0}
\xi_n(\varsigma) = \min\left\{\max\left\{ \ln w_n -\ln \varsigma,0\right\},u_n\right\}
\end{align}
whose graph is shown in Fig. \ref{fig1}.

As seen, the first operation of Algorithm 1 is to compute the quantities $\varsigma_n^\star$ for $n=1,\ldots, 4$ according to step b). Since the condition $\gamma_n \le \sum\nolimits_{i=1}^n u_i$ is satisfied for $n=1,2,\ldots, 4$, the computation of $\varsigma_n^\star$ requires to solve \eqref{100.10} for $n=1,2,\ldots, 4$. Using \eqref{xi_n.0}, we easily obtain:
\begin{align}
\varsigma_1^\star &= e^{\ln w_1 - \gamma_1} = 1.637\\
\varsigma_2^\star &= e^{{\ln w_1 + u_2 - \gamma_2}} = 4.451\\
\varsigma_3^\star &= e^{{\ln w_3 + u_1 + u_2 - \gamma_3}} = 1.196\\
\varsigma_4^\star &= e^{\frac{\ln w_1 + \ln w_3 + u_2 + u_4 - \gamma_4}{2}} = 2.307.
\end{align}
%
%
%
A direct depiction of the above results can be easily obtained by plotting $c_n(\varsigma)$ for $n=1,2,\ldots,4$ as a function of $\varsigma \ge 0$. As shown in Fig. \ref{fig2}, the intersections of curves $c_n(\varsigma)$ with the horizontal lines at $\varsigma = \gamma_n$ yield $\varsigma_n^\star$.

\begin{figure}[t]
\begin{center}
\psfrag{x1}[r][m]{\scriptsize{$\gamma_1 = 0.2$}}
\psfrag{x2}[r][m]{\scriptsize{$\gamma_2 = -2$}}
\psfrag{x3}[r][m]{\scriptsize{$\gamma_3 = 1.1$}}
\psfrag{x4}[c][m]{\scriptsize{$ \gamma_4 = -1.9$\quad\quad\quad\quad}}
\psfrag{x6}[c][m][1][90]{\scriptsize{$\varsigma_1^\star = 1.637$}}
\psfrag{x8}[c][m][1][90]{\scriptsize{$\varsigma_2^\star = 4.451$}}
\psfrag{x7}[c][m][1][90]{\scriptsize{$\varsigma_3^\star = 1.196$}}
\psfrag{x5}[c][m][1][90]{\scriptsize{$\;\varsigma_4^\star = 2.307$}}
\psfrag{x9}[c][m]{\scriptsize{$\varsigma$}}
\psfrag{data1}[c][m]{\scriptsize{$\quad c_1(\varsigma)$}}
\psfrag{data2}[c][m]{\scriptsize{$\quad c_2(\varsigma)$}}
\psfrag{data3}[c][m]{\scriptsize{$\quad c_3(\varsigma)$}}
\psfrag{data4}[c][m]{\scriptsize{$\quad c_4(\varsigma)$}}
\includegraphics[width=0.9\columnwidth]{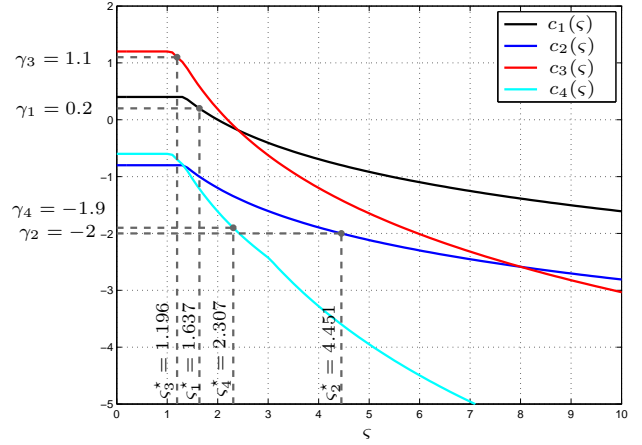}
\end{center}
\caption{Graphical illustration of $c_n(\varsigma)$. Their intersection with the horizontal dashed lines at $\gamma_1 = 0.2$, $\gamma_2 = -2$, $\gamma_3 = 1.1$ and $\gamma_4 = -1.9$ yields respectively $\varsigma_1^\star = 1.637$, $\varsigma_2^\star = 4.451$, $\varsigma_3^\star = 1.196$ and $\varsigma_4^\star = 2.307$. \vspace{-0.3cm}}
\label{fig2}
\end{figure}

Using the above results into \eqref{101.10} and \eqref{102.10} of step c) yields $\mu^\star = 4.451$ and $k^\star = 2$
from which (according to step d)) we obtain 
\begin{align}
\sigma_1^\star=\sigma_2^\star=\mu^\star = 4.451. 
\end{align}
Once the optimal $\sigma_1^\star$ and $\sigma_2^\star$ are computed, Algorithm 1 proceeds solving the following reduced problem:
\begin{align}
  \underset{\{x_3,x_4\}}{\min} \quad &
\sum\limits_{n=3}^4 {w_n}e^{-x_n}\\\nonumber
\text{subject to} \quad & \sum \limits_{n=3}^{j} x_n\le \gamma_j\quad j=3,4\\\nonumber
 & x_n \le u_n \quad n=3,4
\end{align}
with $\gamma_3 = 3.1$ and $\gamma_4 = 0.1$ as obtained from $\gamma_j \leftarrow \gamma_j - \gamma_{k^\star}$
observing that $\gamma_{k^\star}=\gamma_2 = -2$. Since $\gamma_3 > u_3 $, from step b) we have that $\varsigma_3^\star = 0$ while $\varsigma_4^\star$ turns out to be given by 
\begin{align}
\varsigma_4^\star = e^{{\ln w_3 + u_4 - \gamma_4}} = 1.195.
\end{align}
As before, $\varsigma_4^\star$ can be obtained as the intersection of new function 
\begin{align}
c_4(\varsigma) = \sum\limits_{n=3}^4 \xi_n(\varsigma) 
\end{align}
with the horizontal line at $\varsigma = \gamma_4 = 0.1$. Then, from \eqref{101.10} and \eqref{102.10}, we have that 
\begin{align}
\mu^\star = \mathop {\max }\nolimits_{n=3,4 } \varsigma_n^\star= 1.195
\end{align}
and thus $k^\star = 4$. This means that $\sigma_3^\star= \sigma_4^\star=1.195$.

The optimal $x_n^\star$ are eventually obtained as $x_n^\star = \xi_n(\sigma_n^\star)$. This yields $x_1^\star = -0.8$, $x_2^\star = -1.2$, $x_3^\star = 1.9$ and $x_4^\star = -1.8$. As depicted in Fig. \ref{fig1}, the solution $x_n^\star$ corresponds to the interception of $\xi_n(\varsigma)$ with the vertical line at $\varsigma = \sigma_n^\star$.

 \subsection{Water-filling inspired policy}
 {While the charts used in the foregoing example are quite useful, we put forth an alternative interpretation that allows retaining some of the intuition of the water-filling policy. This interpretation is valid for cases in which the optimization variables $\{x_n;\,n=1,2,\ldots,N\}$ can only take non-negative values, which amounts to setting $l_n=0$ for $n=1,2,\ldots,N$. 
 
We start considering the simple case in which a single linear constraint is imposed:
\begin{align}\label{wfl1}
 \underset{\{x_n\}}{\min} \quad &
\sum\limits_{n=1}^N f_n(x_n)\\
\nonumber  \text{subject to} \quad & \sum \limits_{n=1}^{N} x_n\le \rho_N \\
 \nonumber & 0 \le x_n \le u_n \quad  n=1,2,\ldots,N.
\end{align}
Using the results of Theorem 1 and Lemma 3, the solution to \eqref{wfl1} is found to be
\begin{equation}\label{wfl2}
x_n^\star= \xi_n(\sigma^\star)=\min\left\{\max\left\{h_n^{-1}(\sigma_n^\star), 0\right\}, u_n\right\}
\end{equation}
where the values of $\sigma_n^\star$ are obtained through Algorithm 1. Since a single constraint is present in \eqref{10}, then a single iteration is required to compute all the values of $\sigma_n^\star$ for $n=1,2,\ldots,N$. In particular, it turns out that $\sigma_n^\star = \sigma^\star$ for any $n$, with $\sigma^\star$ such that the following condition is satisfied:
\begin{align}\label{wfl4}
\sum\limits_{n=1}^N x_n^\star =\sum\limits_{n=1}^N \xi_n(\sigma^\star)= \rho_N.
\end{align}
%
%

 \begin{figure}[t]
\begin{center}
\includegraphics[width=.55\textwidth]{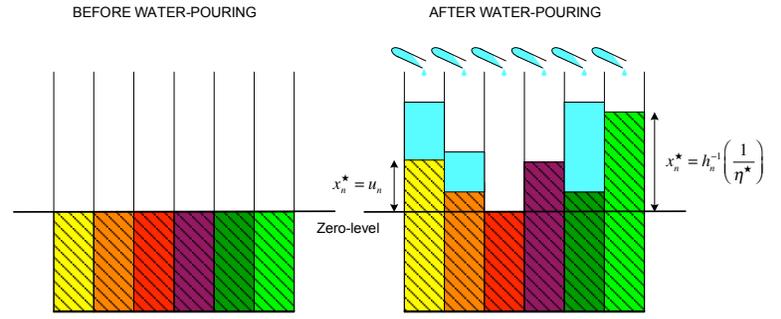}
\end{center}
\caption{Water-filling inspired interpretation of the solutions $x_n^\star$ .}
\label{GWF1}\vspace{-0.4cm}
\end{figure}

Consider now $N$ vessels, which are filled with a proper material (different from vessel to vessel) up to a given level. Think of it as the \textit{zero-level} and assume that it is the same for all vessels, as illustrated in Fig. \ref{GWF1} for $N = 6$. Assume that a certain quantity $\eta$ of water (measured in proper units) is poured into each vessel and let each material be able to first absorb it and then to expand accordingly up to a certain level. In particular, assume that the behaviour of material $n$ is regulated by $\xi_n(\varsigma)$ with $\varsigma= 1/\eta$. More precisely, $\xi_n(\varsigma)$ is the difference between the new level of material $n$ and the zero-level. From \eqref{xi_n}, it easily follows that the expansion starts only when $\eta$ reaches the level $\eta=1/h_n(0)$ while it stops when $\eta=1/h_n(u_n)$, corresponding to a maximum expansion of $\xi_n(\varsigma)=u_n$. This means that additional water beyond the quantity $1/h_n(u_n)$ does not produce any further expansion - it is simply accumulated in vessel $n$ above the level $u_n$ as depicted in Fig. \ref{GWF1}.

Using \eqref{wfl2} and \eqref{wfl4}, the solutions $\{x_n^\star\}$ to \eqref{wfl1} can thus be interpreted as obtained trough the following procedure, which is reminiscent of the water-filling policy.

\begin{enumerate}
\item Consider $N$ vessels;
\item Assume vessel $n$ is filled with a proper material up to a certain  zero-level (the same for each vessel);
\item Let the behaviour of material $n$ be regulated by $\xi_n$;
\item Compute $\sigma^\star$ through \eqref{wfl4}; 
\item Poor the same quantity $\eta^\star=1/\sigma^\star$ of water into each vessel; 
\item The material height over the zero-level in vessel $n$ gives $x_n^\star$.
\end{enumerate}



The extension of the above water-filling interpretation to the general form in \eqref{1} is straightforward. Assume that the $j$th iteration is considered. Then, Algorithm 1 proceeds as follows.

\begin{enumerate}
\item Consider $N_j$ vessels with indices $n=j+1,\ldots,N$;
\item Assume vessel $n$ is filled with a proper material up to a certain  zero-level (the same for each vessel);
\item Let the behaviour of material $n$ be regulated by $\xi_n$;
\item Compute $\mu^\star$ and $k^\star$ through \eqref{101.10} and \eqref{102.10}; 
\item Poor the same quantity $\eta^\star=1/\mu^\star$ of water into  vessels $n = j+1,\ldots,k^\star$; 
\item The material height over the zero-level gives $x_n^\star$ for $n = j+1,\ldots,k^\star$.
\end{enumerate}

\begin{remark}
Observe that the speed by which material $n$ expands itself depends on  $\xi_n^\prime $ defined as the first derivative of $\xi_n$ with respect to $\eta = 1/\varsigma$. It can be easily shown that 
\begin{align}\label{pippo}
\xi_n^\prime = \frac{1}{\eta^2 f_n^{\prime\prime}\left(h_n^{-1}(1/\eta)\right)}
\end{align}
from which it follows that the rate of growth is inversely proportional to the second derivative of 
$f_n$ evaluated at $h_n^{-1}(1/\eta)$. 
\end{remark}
}


\vspace{-0.2cm}
\section{Particularization to power allocation problems}\label{examples}
In the following, we show how some power allocation problems in signal processing and communications can be put in the form of \eqref{1}, and thus can be solved with the generalized algorithm illustrated above\footnote{Due to the considerable amount of works in this field, our exposition will be necessarily incomplete and will reflect the subjective tastes and interests of the authors. To compensate for this partiality, we refer the interested reader to the list of references for an entree into the extensive literature on this subject.}. 
%

\vspace{-0.3cm}
\subsection{{Classical water-filling and cave-filling policies}}

\begin{figure}[t]
\begin{center}
\includegraphics[width=.4\textwidth]{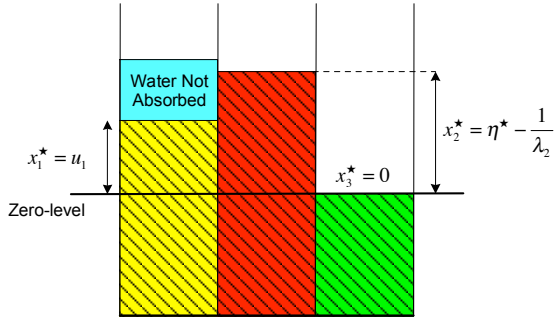}
\end{center}
\caption{Illustration of the water-filling inspired policy for problem \eqref{10} when $N=3$.}
\label{Fig8_new}
\end{figure}

Consider the classical problem of allocating a certain amount of power $P$ among a bank of non-interfering channels to maximize the capacity. This problem can be mathematically formulated as follows:
\begin{align}\label{10}
   \underset{\{x_n\}}{\max} \quad &
\sum\limits_{n=1}^N \log(1+\lambda_n x_n)\\\nonumber
\text{subject to} \quad & \sum \limits_{n=1}^{N} x_n\le P \quad \\\nonumber
 & 0 \le x_n \le u_n \quad n=1,2,\ldots,N
\end{align}
where $x_n$ represents the transmit power allocated over the $n$th channel of gain $\lambda_n$ whereas $\log(1+\lambda_n x_n)$ gives the capacity of the $n$th channel. Clearly, we assume that $\sum_{n=1}^{N} u_n > P$, otherwise \eqref{10} has the trivial solution $x_n^\star=u_n$.

The above problem can be put in the same form of \eqref{1mod} setting $f_n(x_n) = -\log (1+\lambda_n x_n)$, 
$l_n=0$ $\forall n$, $\mathcal{L}=\{N\}$ and $\rho_N = P$. Observing that 
\begin{align}
h_n^{-1}(\varsigma) = \frac{1}{\varsigma} - \frac{1}{\lambda_n} 
\end{align}
from \eqref{x_n^star_2} one gets
\begin{equation}\label{cavefilling}
x_n^\star=\min\left\{\max\left\{\dfrac{1}{\sigma^\star} - \dfrac{1}{\lambda_n},0\right\},u_n\right\}
\end{equation}
with $\sigma^\star$ such that 
\begin{align}\label{constraint}
\sum\limits_{n=1}^N x_n^\star =\sum\limits_{n=1}^N \xi_n(\sigma^\star)= P.
\end{align}
{Using the water-filling policy illustrated in Section \ref{graphical_interpretation}, the solutions in \eqref{cavefilling} have the visual interpretation shown in Fig.~\ref{Fig8_new}, where we have assumed $N=3$ and set $\eta^{\star}=1/\sigma^\star$.  The material inside the $n$th vessel starts expanding when the quantity of water $\eta$ poured in the vessel equals $1/\lambda_n$. Due to the particular form of $f_n$, the expansion follows the linear law $\xi_n(1/\eta)= \eta-1/\lambda_n$ as long as $\eta \le u_n+1/\lambda_n$. After that, water is no more absorbed and the expansion stops. The additional water is accumulated in the vessel above the maximum level of the material. As shown in Fig.~\ref{Fig8_new}, this is precisely what happens with the yellow material in vessel $1$. On the other hand, we have that $\eta^\star-1/\lambda_2 < u_2$ and thus no water is accumulated on the top of the red material in vessel $2$. Finally, the green material in vessel $3$ is such that no expansion occurs since $\eta^\star<1/\lambda_3$.}

\begin{figure}[t]
\begin{center}
\includegraphics[width=.4\textwidth]{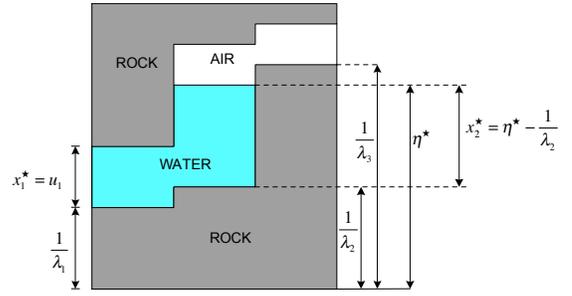}
\end{center}
\caption{Illustration of the cave-filling policy for problem \eqref{10} when $N=3$.}
\label{Fig8}
\end{figure}

An alternative visual interpretation of \eqref{cavefilling} (commonly used in the literature) is given in Fig.~\ref{Fig8}, where ${1}/{\lambda_n}$ and $u_n + {1}/{\lambda_n}$ are viewed as the ground and the ceiling levels of patch $n$, respectively. In this case, the solution is computed as follows. We start by flooding the region with water to a level $\eta$. The total amount of water used is then given by
\begin{align}
\sum\limits_{n=1}^N \min\left\{\max\left\{\eta - \frac{1}{\lambda_n}, 0\right\},u_n\right\}.
\end{align}
The flood level is increased until a total amount of water equal to $P$ is used. The depth of water inside patch $n$ gives $x_n^\star$. This solution method is known as cave-filling due to its specific physical meaning. Clearly, if $u_n = +\infty$ for any $n$ in \eqref{10} then $x_n^\star$ reduces to
\begin{equation}\label{}
x_n^\star=\max\left\{\dfrac{1}{\sigma^\star} - \dfrac{1}{\lambda_n},0\right\}.
\end{equation}
which is the well-known and classical water-filling solution.

A problem whose solution has the same visual interpretation of Fig. \ref{Fig8} is considered also in \cite{Gao2008} (see problem (21)) in which the authors design the optimal training sequences for channel estimation in multi-hop transmissions using decode-and-forward protocols.

\vspace{-0.3cm}
\subsection{{General water-filling policies}}

\begin{figure}[t]
\begin{center}
\includegraphics[width=.4\textwidth]{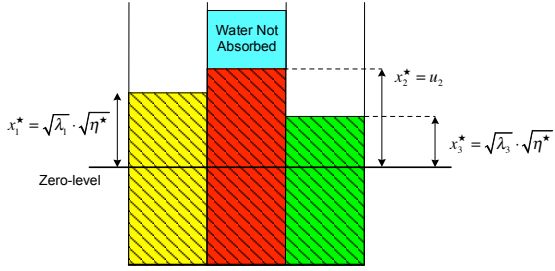}
\end{center}
\caption{Illustration of the water-filling inspired policy for problem \eqref{60} .}
\label{Fig11}
\end{figure}


Consider now the following problem:
\begin{align}\label{60}
   \underset{\{x_n\}}{\min} \quad &
\sum\limits_{n=1}^N \frac{\lambda_n}{x_n}\\\nonumber
\text{subject to} \quad & \sum \limits_{{n=1}}^{{j}} x_n\le \rho_j \quad j=1,2,\ldots,N\\\nonumber
 & 0 \le x_n \le 1 \quad n=1,2,\ldots,N
\end{align}
where $\{\lambda_n > 0\}$ are positive parameters. This problem is  considered in \cite{Sanguinetti2012} in the context of linear transceiver design architectures for MIMO networks with a single non-regenerative relay. It also appears in \cite{PalomarQoS2004} where the authors deal with the linear transceiver design problem in MIMO point-to-point networks to minimize the power consumption while satisfying specific QoS constraints on the mean-square-errors (MSEs). A similar instance can also be found in \cite{Palomar2003} and corresponds to the minimization of the weighted arithmetic mean of the MSEs in a multicarrier MIMO system with a total power constraint. All the above examples could in principle be solved with (specifically designed) multi-level water-filling algorithms \cite{Palomar2005}. Easy reformulations allow to use the more general Algorithm 1 as shown next for problem \eqref{60}.

%
%
 Setting
$f_n(x_n) = {\lambda_n}/{x_n}$
and letting $l_n=0$ and $u_n = 1$, $\forall n$, it is easily seen that \eqref{60} has the same form as \eqref{1}. Then, one gets $h_n(x_n) = {\lambda_n}/{x_n^2}$
and $h_n^{-1}(\varsigma) = \sqrt{{\lambda_n}/{\varsigma}}.$
The solution to \eqref{60} is given by
\begin{align}\label{65}
x_n^\star=\min\left\{\max\left\{\sqrt{\dfrac{\lambda_n}{\sigma_{n}^\star}},0 \right\},1\right\}
\end{align}
where $\{\sigma_{1}^\star, \sigma_{2}^\star,\ldots,\sigma_{N}^\star\}$ are computed through Algorithm 1 and take the form \eqref{rem3_1} with $\{\mu_{1}^{\star},\mu_{2}^{\star},\ldots,\mu_{{L-1}}^{\star}\}$ such that
\begin{align}
\sum\limits_{n= k_{j-1}^{\star}+1}^{k_{j}^{\star}} x_n^\star = \rho_{k_{j}^{\star}}- \rho_{k_{j-1}^{\star}} \quad j=1,2,\ldots,{L-1}.
\end{align}
{According to Remark 4}, if $\mu_{L}^{\star}$ is greater than $0$ then
\begin{align}
\sum\limits_{n= k_{L-1}^{\star}+1}^{N} x_n^\star = \rho_{N}- \rho_{k^{\star}_{L-1}}
\end{align}
otherwise when $\mu_{L}^{\star} = 0$ one gets
\begin{equation}
\label{ }
x_n^\star=1 \quad \quad n=k_{L-1}^{\star}+1,\ldots,N.
\end{equation}
{The solutions $x_n^\star$ in \eqref{65} can be thought as obtained through the water-filling policy illustrated in Section \ref{graphical_interpretation} in which the expansion of material $n$ is regulated by the square-root law $\xi_n(1/\eta) = \sqrt{\lambda_n\eta}$ with rate of growth given by
\begin{equation}
\label{ }
\xi_n(1/\eta)=\sqrt{\frac{\lambda_n}{\eta}},
\end{equation}
according to \eqref{pippo}. This is illustrated in Fig. \ref{Fig11} wherein we consider the first iteration of Algorithm 1 under the assumption that  $k_1^\star = 3$ and $\lambda_3<\lambda_1<\lambda_2$. As expected, the level of the red material in the $2$nd vessel is higher than the others.}

%
%
%

\vspace{-0.3cm}
\subsection{Some other examples}

Consider now the following problem:
\begin{align}\label{70}
   \underset{\{x_n\}}{\min} \quad &
\sum\limits_{n=1}^N {\lambda_n}{e^{-x_n}}\\\nonumber
\text{subject to} \quad & \sum \limits_{{n=1}}^{{j}} x_n\le \rho_j \quad j=1,2,\ldots,N\\\nonumber
 & x_n \le 0 \quad n=1,2,\ldots,N.
\end{align}
The above problem arises in \cite{Jiang2006} where the authors deal with the power minimization in MIMO point-to-point networks with non-linear architectures at the transmitter or at the receiver. A similar problem arises when two-hop MIMO networks with a single amplify-and-forward relay are considered \cite{Sanguinetti2012}.
The solution of \eqref{70} has the form
\begin{equation}\label{72}
x_n^\star=\min\left\{\max\left\{\log\left(\dfrac{\lambda_n}{\sigma_n^\star}\right),0 \right\},u_n\right\}
\end{equation}
where the quantities $\sigma_{n}^\star$ are given by \eqref{rem3_1}.



Another instance of \eqref{1} arises in connection with the computation of the optimal power allocation for the maximization of the instantaneous received signal-to-noise ratio in amplify-and-forward multi-hop transmissions under short-term power constraints \cite{Farhadi09}. Denoting by $N$ the total number of hops, the problem can be mathematically formalized as follows \cite{Farhadi09}
\begin{align}\label{600}
   \underset{\{x_n\}}{\max} \quad &
\left(\prod_{n=1}^N \left(1+\frac{1}{x_n\lambda_n}\right)-1\right)^{-1}\\\nonumber
\text{subject to} \quad & \sum \limits_{n=1}^{N} x_n\le P \\\nonumber
 & 0 \le x_n \le u_n \quad n=1,2,\ldots,N.
\end{align}
where $x_n$ represents the power allocated over the $n$th hop and $P$ denotes the available power. In addition, $\lambda_n$ is the channel gain over the $n$th hop. The above problem can be equivalently reformulated as follows
\begin{align}\label{601}
   \underset{\{x_n\}}{\min} \quad &
\sum_{n=1}^N \log \left(1+\frac{1}{x_n\lambda_n}\right)\\\nonumber
\text{subject to} \quad & \sum \limits_{n=1}^{N} x_n\le P \\\nonumber
 & 0 \le x_n \le u_n \quad n=1,2,\ldots,N
\end{align}
from which it is clear that it is in the same form as \eqref{1mod} with
\begin{align}
f_n(x_n) = \log \left(1+\frac{1}{x_n\lambda_n}\right)
\end{align}
$\mathcal L = \{N\}$ and $\rho_N = P$. Then, 
\begin{align}\label{605}
h_n^{-1}(\varsigma) = \frac{\sqrt{1+\frac{4\lambda_n}{\varsigma}}-1}{2\lambda_n}.
\end{align}
It is assumed $\sum_{n=1}^{N} u_n > P$, otherwise \eqref{601} has the trivial solution $x_n^\star=p_n$. Using \eqref{605} into \eqref{x_n^star} yields
\begin{equation}\label{x_n.3}
x_n^\star=\min\left\{\max\left\{\frac{1}{2\lambda_n}\Big({\sqrt{1+\dfrac{4\lambda_n}{\sigma^\star}}-1}\Big),0 \right\},u_n\right\}
\end{equation}
with $\sigma^\star$ such that $\sum\nolimits_{n=1}^N x_n^\star = P$.

\section{Conclusions}\label{conclusions}

{An iterative algorithm has been proposed to compute the solution of separable convex optimization problems with a set of linear and box constraints. The proposed solution operates through a two layer architecture, which has a simple graphical water-filling inspired interpretation. The outer layer requires at most $N-1$ steps with $N$ being the number of linear constraints whereas the number of iterations of the inner layer depends on the complexity of solving a set of (possibly) non-linear equations. If solvable in closed form, then the computational burden of the inner layer is negligible. The problem under investigation is particularly interesting since a large number of existing (and likely future) power allocation problems in signal processing and communications can be reformulated as instances of its general form, and thus can be solved with the proposed algorithm without the need of developing specific solutions for each of them.}

\section*{Appendix A \\ Proof of Lemmas 1 and 2}
We start considering case {\bf{a}}). Without loss of generality, we concentrate on $f_1$, which is assumed monotonically increasing in $[l_1,u_1]$, and aim at proving that $x^\star_1 = l_1$. {We start denoting by $\mathcal{S}(x_1)$ the feasible set of $x_2,x_3,\ldots,x_N$ for a given $x_1 \in [l_1,u_1]$. Mathematically, $\mathcal{S}(x_1)$ is such that
\begin{align}\label{Opt_Prob_Red.1}
&\sum\limits_{n=2}^{j}x_n \le \rho_{n}-x_1 \quad j=2,\ldots,N \\\nonumber
   & l_{n} \le x_n \le u_{n} \quad n=2,\ldots,N.
 \end{align}
 Clearly, we have that $\mathcal{S}(x_1) \subseteq \mathcal{S}(l_1)$
 for any $x_1 \, \in \, (l_1,u_1]$. For notational convenience, we also define $F(x_1)$ as
   \begin{align}\label{Opt_Prob_Red.2}
   F(x_1) = \underset{\{x_2,x_3,\ldots,x_N\} \in \mathcal{S}(x_1)}{\min} \quad
  \sum\limits_{n=2}^{N}f_{n}({x_n}).
 \end{align}
Observe now that the optimal value $x^{\star}_{1}$ is such that $f_1(x_1) + F(x_1)$ is minimized. To this end, we recall that: $\mathbf{i}$) $ \; f_1(l_1) < f_1(x_1)$ since $f_1$ is strictly increasing in $[l_1, u_1]$; $\mathbf{ii}$) $\;F_1(l_1) \le F_1(x_1)$ since $\mathcal{S}(x_1) \subseteq \mathcal{S}(l_1)$ for any $x_1 \, \in \, (l_1,u_1]$. Therefore, it easily follows that $f_1(l_1) + F(l_1) < f_1(x_1) + F(x_1)$ for any $x_1 \, \in \, (l_1,u_1]$, which proves that $x^{\star}_{1}= l_1$. The same result can easily be extended to a generic $x_n$ with $n \ne 1$ using similar arguments. This proves Lemma 1.}

{Consider now case {\bf{b}}) and assume that there exists a point $z_n$ in $(l_n,u_n)$ such that $f_n^{'} (z_n) = 0 $ with $f_n'(x_n)<0\; \forall x_n \in (l_n, z_n)$ and $f_n'(x_n)>0\; \forall x_n \in (z_n, u_n)$. We aim at proving that $x^\star_n \in [l_n,z_n]$}

{Since $f_n'(x_n)>0\; \forall x_n \in (z_n, u_n)$, then $f_n(x_n)$ is monotonically increasing in $[z_n,u_n]$. Consequently, by Lemma 1 it follows that $x^\star_n$ cannot be greater than $z_n$. This amounts to saying that $x^{\star}_n$ must belong to the interval $[l_n,z_n]$, as stated in \eqref{solB}.}

Finally, for case {\bf{c}}) nothing can be said a priori apart for that the solution $x^{\star}_n$ lies in interval $[l_n,u_{n}]$ as required by the box constraints in \eqref{1}.

\section*{Appendix B \\ Proof of Theorem 1}

We begin by writing the Karush-Kuhn-Tucker (KKT) conditions for the convex problem $(\mathcal{P})$:
\begin{equation}\label{KKT1}
 -h_{n}(x_n)+\sum\limits_{j=n}^{N}\lambda_j+\nu_n-\kappa_n=0 \quad \quad n=1,\ldots,N
\end{equation}
\begin{equation}\label{KKT2}
 0 \le \lambda_n \; \bot \; \Big(\sum\limits_{j=1}^{n}x_j-\rho_n\Big) \le 0 \quad \quad n=1,\ldots,N
\end{equation}
\begin{equation}\label{KKT3.1}
 0 \le \nu_n \; \bot \; \left(x_n-u_n\right) \le 0 \quad \quad n=1,\ldots,N
\end{equation}
\begin{equation}\label{KKT4.1}
 0\le \kappa_n \; \bot \; \left(x_n-l_n\right) \ge 0 \quad \quad n=1,\ldots,N
\end{equation}
where $h_{n}(x) = -f^{\prime}_{n}(x)$.
Letting $\sigma_n =\sum\nolimits_{j=n}^{N}\lambda_j$ and $\sigma_{N+1} = 0$, we may rewrite \eqref{KKT1} -- \eqref{KKT4.1} in the following equivalent form:
\begin{equation}\label{KKT1eq}
 -h_{n}(x_n)+\sigma_n+\nu_n-\kappa_n=0 \quad \quad n=1,\ldots,N
\end{equation}
\begin{equation}\label{KKT2eq}
 \sigma_n \ge 0 \quad \quad \sigma_{N+1} = 0
 \end{equation}
 \begin{equation}\label{KKT3eq}
 0 \le (\sigma_n-\sigma_{n+1}) \; \bot\; \Big(\sum\limits_{j=1}^{n}x_j-\rho_n\Big) \le 0 \quad n=1,\ldots,N
\end{equation}
\begin{equation}\label{KKT3}
0 \le \nu_n \; \bot \; \left(x_n-u_n\right) \le 0 \quad \quad n=1,\ldots,N
\end{equation}
\begin{equation}\label{KKT4}
 0\le \kappa_n \; \bot \; \left(x_n-l_n\right) \ge 0 \quad \quad n=1,\ldots,N.
\end{equation}

Since ($\mathcal{P}$) is convex, solving the KKT conditions is equivalent to solving ($\mathcal{P}$). {Accordingly, we let $x^{\star}_{n}$, $\nu^{\star}_{n}$, $\kappa^{\star}_{n}$ and $\sigma_n^{\star}$ to denote the solution of \eqref{KKT1eq} -- \eqref{KKT4} for $n=1,\ldots,N$. In the next, it is shown that $x^{\star}_{n}$, $\nu^{\star}_{n}$ and $\kappa^{\star}_{n}$ are given by}
\begin{equation}\label{xsol}
 x_{n}^{\star} = \xi_{n}(\sigma_{n}^{\star})
\end{equation}
\begin{equation}\label{nisol}
 \nu_{n}^{\star} = \max \{h_{n}(x_{n}^{\star})-\sigma_{n}^{\star},0 \}
\end{equation}
\begin{equation}\label{kappasol}
 \kappa_{n}^{\star} = \max \{\sigma_{n}^{\star} -h_{n}(x_{n}^{\star}),0 \}
\end{equation}
where 
$\xi_n(\sigma_{n}^{\star})$ is computed as in \eqref{xi_n} with $\varsigma=\sigma_n^{\star}$:
\begin{equation}\label{csi_n}
 \xi_{n}(\sigma_{n}^{\star}) = \left\{\begin{array}{cl}
             u_{n} & 0 \le \sigma_{n}^{\star} < h_{n}(u_n)\\ \\
             h^{-1}_{n}(\sigma_{n}^{\star}) & h_{n}(u_n) \le \sigma_{n}^{\star} < h_{n}(l_n) \\ \\
             l_n & h_{n}(l_n)\le \sigma_{n}^{\star}. \\
            \end{array}
 \right.
\end{equation}
{The following three cases are considered separately:} $\mathbf{a}$) $x_{n}^{\star}=l_n$; $\mathbf{b}$) $l_n < x_{n}^{\star} < u_n$; $\mathbf{c}$) $x_{n}^{\star} = u_n$.

Case $\mathbf{a}$) If $x_{n}^{\star}=l_n$ then from \eqref{KKT3} it immediately follows $\nu_{n}^{\star}=0$ whereas \eqref{KKT1eq} reduces to $-h_{n}(x_{n}^{\star})+\sigma_n^{\star}=\kappa_n^{\star}$,
from which using \eqref{KKT4} we get
\begin{equation}\label{pluto_1}
-h_{n}(x_{n}^{\star})+\sigma_n^{\star}\ge 0
\end{equation}
or, equivalently, $    \sigma_n^{\star} \ge h_{n}(x_{n}^{\star}) = h_{n}(l_n).$
{Using the above result into \eqref{csi_n} yields
\begin{equation}\label{}
 x_n^{\star}=l_n=\xi_n(\sigma_n^{\star})
\end{equation}
as stated in \eqref{xsol}.}
From the above results, it also follows that:
   \begin{align}
    \nu_{n}^{\star}&=0=\max\{h_{n}(x_{n}^{\star})-\sigma_n^{\star},0\}\\
    \kappa_{n}^{\star}&=\sigma_n^{\star}-h_{n}(x_{n}^{\star})=\max\{\sigma_n^{\star}-h_{n}(x_{n}^{\star}),0\}
   \end{align}
as given in \eqref{nisol} and \eqref{kappasol}, respectively.

Case $\mathbf{b}$) From \eqref{KKT3} and \eqref{KKT4} we obtain $\nu_{n}^{\star}=0$ and $\kappa_{n}^{\star}=0$ so that \eqref{KKT1eq} reduces to
\begin{equation}\label{KKT1eq_ii}
{- h_{n}(x_{n}^{\star})}+\sigma_n^{\star}=0
\end{equation}
from which we have that $x_{n}^{\star}=-h_n^{-1}(\sigma_n^{\star})$. Since in this case $l_n < x_{n}^{\star} < u_n$, then
\begin{equation}\label{sigma_n_star_ii}
h_{n}(u_{n}) < \sigma_n^{\star} < h_{n}(l_{n})
\end{equation}
so that we obtain
\begin{equation}\label{}
 x_n^{\star}=h_n^{-1}(\sigma_n^{\star})=\xi_n(\sigma_n^{\star}).
\end{equation}
Also, taking \eqref{KKT1eq_ii} into account yields
\begin{align}
\nu_{n}^{\star}&=0=\max\{h_{n}(x_{n}^{\star})-\sigma_n^{\star},0\}\\
\kappa_{n}^{\star}&=0=\max\{\sigma_n^{\star}-h_{n}(x_{n}^{\star}),0\}.
\end{align}

Case $\mathbf{c}$) In this case, from \eqref{KKT4} one gets $\kappa_{n}^{\star}=0$ whereas \eqref{KKT1eq} reduces to $\nu_n^{\star} = h_{n}(x_{n}^{\star})-\sigma_n^{\star}$. Since \eqref{KKT3} is satisfied for $\nu_n^{\star} \ge 0$, then $h_{n}(x_{n}^{\star})-\sigma_n^{\star} \ge 0$ or, equivalently,
\begin{equation}\label{sigma_n_star_iii}
 \sigma_n^{\star} \le h_{n}(x_{n}^{\star}) =h_{n}(u_{n}).
\end{equation}
Accordingly, we can write
\begin{equation}\label{}
 x_n^{\star}=u_n=\xi_n(\sigma_n^{\star})
\end{equation}
\begin{equation}
\nu_{n}^{\star}=h_{n}(x_{n}^{\star})-\sigma_n^{\star}=\max\{h_{n}(x_{n}^{\star})-\sigma_n^{\star},0\}
\end{equation}
and
\begin{equation}
\kappa_{n}^{\star}=0=\max\{\sigma_n^{\star}-h_{n}(x_{n}^{\star}),0\}.
\end{equation}

Using all the above results together, \eqref{xsol} -- \eqref{kappasol} easily follow from which it is seen that $x^{\star}_n$ depends solely on $\sigma_n^{\star}$. The latter must be chosen so as to satisfy \eqref{KKT2eq} and \eqref{KKT3eq}.

%
\section*{Appendix C \\ Proof of Lemma 3}

\begin{algorithm}[t]
\caption{Equivalent form of {\bf{Algorithm 1}}.}

\begin{enumerate}
\item Set $j=1$, $k_0^{\star} = 0$ and $\gamma_n = \rho_n$ for every $n$.
\item {{\bf{While}}} $k_{j-1}^{\star}< N$
\begin{enumerate}
\item Set $\mathcal N_{k_j^{\star}}= \{k_{j-1}^{\star}+1,k_{j-1}^{\star}+2, \ldots, N\}$.
\item For every $n$ in $\mathcal N_{k_j^{\star}}$.

\begin{enumerate}
\item If $\gamma_n < \sum\nolimits_{i=k_{j-1}^{\star}+1}^n u_i$ then compute $\varsigma_{n,j}^\star$ as the solution of
\begin{align}\label{100}
c_n(\varsigma;j)=\sum\limits_{i=k_{j-1}^{\star}+1}^n \xi_i(\varsigma) = \gamma_n
\end{align}
for $\varsigma$.
\item If $\gamma_n \ge \sum\nolimits_{i=k_{j-1}^{\star}+1}^n u_i$ then set
\begin{align}\label{101.0120}
\varsigma_{n,j}^\star = 0.
\end{align}
\end{enumerate}
\item Evaluate
\begin{align}\label{101}
\mu_j^\star= \underset{n\in \mathcal{N}_{k_j^{\star}}}{\max} \quad \varsigma_{n,j}^\star
\end{align}
and
\begin{align}\label{102}
k_j^\star= \underset{n\in \mathcal{N}_{k_j^\star}}{ \max} \left\{n | \varsigma_{n,j}^\star=\mu_j^\star\right\}.
\end{align}
\item Set
\begin{align}\label{104}
\sigma_n^\star \leftarrow \mu_j^\star \quad \text{for}\;\; k_{j-1}^\star < n \le k_j^\star.
\end{align}
\item Set $\gamma_n \leftarrow \rho_n - \rho_{k_{j-1}^\star}$
for $k_{j-1}^\star < n \le N$.
\item Set $j\leftarrow j + 1$.
\end{enumerate}
\end{enumerate}
\end{algorithm}

For the sake of clarity, the steps of Algorithm 1 are put in the equivalent forms illustrated in {\bf{Algorithm 2}} in which basically some indices and equations are introduced or reformulated in order to ease understanding of the mathematical arguments and steps reported below.

{As seen, the $j$th iteration of Algorithm 2 computes the real parameter $\mu_j^{\star}$ and the integer $k_j^{\star} \in \{k_{j-1}^{\star}+1,\ldots,N\}$ through \eqref{101} and \eqref{102}, respectively. The latter are then used in \eqref{104} to obtain $\sigma_n^{\star}$ for $n=k_{j-1}^{\star}+1,\ldots,k_j^{\star}$:}
\begin{equation}\label{sigmasolkr}
 \sigma_{n}^{\star} = \mu_j^{\star} \quad \text{for} \;\; k_{j-1}^\star < n \le k_{j}^\star.
\end{equation}
In the next, it is shown that the quantities $\sigma_{n}^{\star}$ given by \eqref{sigmasolkr} satisfy \eqref{KKT2eq} and \eqref{KKT3eq}.

{We start proving that $\sigma_{n}^{\star} \ge 0$
as required in \eqref{KKT2eq}. Since the domain of the function $c_n(\varsigma;j)$ in \eqref{100} is the interval $[0,+\infty)$ then the solution $\varsigma_{n,j}^\star$ of $c_n(\varsigma;j)=\rho_n-\rho_{k_{j-1}^{\star}}$ is non-negative or it does not exist. In the latter case, Algorithm 2 sets $\varsigma_{n,j}^\star=0$ (according to \eqref{101.0120}) and hence $\varsigma_{n,j}^\star \ge 0$ in any case. This means that $\mu_j^{\star}$, as computed through \eqref{101}, is non-negative and, consequently, $\sigma_{n}^{\star}$ is non-negative as well.}

{To proceed further, we now show that
\begin{equation}\label{KKT3eq_AxC}
 \sigma_n^{\star}-\sigma_{n+1}^{\star} \ge 0
\end{equation}
for $n=1,\ldots,N$ as required in \eqref{KKT3eq}. To this end, we start observing that $\forall j$
\begin{align}\label{sn_inc_1}
\sigma_n^{\star}-\sigma_{n+1}^{\star} & =0 \quad \text{for} \;\; k_{j-1}^{\star} < n < k_{j}^{\star}
\end{align}
as immediately follows from \eqref{sigmasolkr}. On the other hand, for $n=k_{j}^{\star}$ one has
\begin{align}\label{sn_inc_2}
\sigma_{k_{j}^{\star}}^{\star}-\sigma_{k_{j}^{\star}+1}^{\star} =\mu_{j}^{\star}-\mu_{j+1}^{\star}.
\end{align}
From \eqref{sn_inc_1} and \eqref{sn_inc_2}, it clearly follows that to prove \eqref{KKT3eq_AxC} it suffices to show that $\mu_{j}^{\star} \ge \mu_{j+1}^{\star}$. To see how this comes about, we start observing that since each $\xi_n(\varsigma)$ in \eqref{xi_n} is non-increasing then $c_n(\varsigma;j)$ in \eqref{100} is non-increasing as well, so that we may write
\begin{equation}\label{1000}
c_n(\mu_j^{\star};j) \le c_n(\varsigma_{n,j}^\star;j) = \rho_n-\rho_{k_{j-1}^{\star}} \quad n=k_{j-1}^{\star}+1,\ldots,N
 \end{equation}
where we have taken into account that by definition $\mu_j^{\star} \ge \varsigma_{n,j}^\star$ (as it follows from \eqref{101}). In particular, \eqref{1000} is satisfied with equality for $n=k_{j}^{\star}$ whereas it is a strict inequality for $n>k_{j}^{\star}$. Indeed, it cannot exist an index $\bar{k} > k_{j}^{\star}$ such that $c_{\bar{k}}(\mu_j^{\star};j) = \rho_{\bar{k}} - \rho_{k_{j-1}^{\star}}$
because this would mean that $\mu_j^{\star}$ is solution of both $c_{\bar{k}}(\varsigma;j)=\rho_{\bar{k}}-\rho_{k_{j-1}^{\star}}$ and $c_{k_{j}^{\star}}(\varsigma;j)=\rho_{k_{j}^{\star}}-\rho_{k_{j-1}^{\star}}$. If that is the case, in applying \eqref{102} at the $j$th step $\bar{k}$ would have been chosen instead of $k_{j}^{\star}$.

Based on the above results, setting $n=k_{j+1}^{\star}>k_{j}^{\star}$ into \eqref{1000} yields
\begin{equation}\label{1022}
c_{k_{j+1}^{\star}}(\mu_j^{\star};j) < \rho_{k_{j+1}^{\star}} - \rho_{k_{j-1}^{\star}}.
 \end{equation}
Also, a close inspection of \eqref{100} reveals that for $n> k_j^{\star}$ $c_{n}(\varsigma;j)$ can be rewritten as follows
  \begin{equation}
c_{n}(\varsigma;j) = c_{k_j^{\star}}(\varsigma;j) + c_{n}(\varsigma;j+1)
 \end{equation}
from which setting $n=k_{j+1}^{\star}$ we obtain
  \begin{equation}\label{1010}
c_{k_{j+1}^{\star}}(\varsigma;j) = c_{k_j^{\star}}(\varsigma;j) + c_{k_{j+1}^{\star}}(\varsigma;j+1).
 \end{equation}
Replacing $\varsigma$ with $\mu_j^{\star}$ in \eqref{1010} yields
  \begin{equation}\label{1011}
c_{k_{j+1}^{\star}}(\mu_j^{\star};j) = \rho_{k_{j}^{\star}} - \rho_{k_{j-1}^{\star}} + c_{k_{j+1}^{\star}}(\mu_j^{\star};j+1)
 \end{equation}
 where we have taken into account that
  \begin{equation}
c_{k_j^{\star}}(\mu_j^{\star};j) = \rho_{k_{j}^{\star}} - \rho_{k_{j-1}^{\star}}
 \end{equation}
 as it easily follows from the definition of $c_{n}(\varsigma;j)$ in \eqref{100} and from those of $\mu_j^{\star}$ and $k_j^{\star}$ in \eqref{101} and \eqref{102}.}

{Using \eqref{1022} with \eqref{1011} leads to
 \begin{equation}\label{1020}
 c_{k_{j+1}^{\star}}(\mu_j^{\star};j+1) < \rho_{k_{j+1}^{\star}} - \rho_{k_{j}^{\star}}
 \end{equation}
from which recalling that
  \begin{equation}
c_{k_{j+1}^{\star}}(\mu_{j+1}^{\star};j+1) = \rho_{k_{j+1}^{\star}} - \rho_{k_{j}^{\star}}
 \end{equation}
we obtain
  \begin{equation}
c_{k_{j+1}^{\star}}(\mu_j^{\star};j+1) < c_{k_{j+1}^{\star}}(\mu_{j+1}^{\star};j+1).
 \end{equation}
Since the functions $c_{n}(\varsigma;j)$ are non-increasing, from the above inequality we eventually obtain $\mu_{j} ^{\star}> \mu_{j+1}^{\star}$
from which using \eqref{sn_inc_2} we have that $\sigma_{k_{j}^{\star}}^{\star}-\sigma_{k_{j}^{\star}+1}^{\star}>0$. Accordingly, from \eqref{sn_inc_1} it follows that $\sigma_n^{\star}-\sigma_{n+1}^{\star} \ge 0$ as required by \eqref{KKT3eq}.}


{We proceed showing that
  \begin{align}\label{0010}
\sum\limits_{i=1}^{n}x^\star_i-\rho_n \le 0 \quad n=1,\ldots,N.
\end{align}
For this purpose, observe that for {$k_{j-1}^{\star} <n \le k_{j}^{\star}$} we have that 
\begin{align}\label{} \nonumber
 \sum\limits_{i=1}^{n}x_i^{\star}& =\rho_{k_{j-1}^{\star}} +\sum\limits_{i= k^\star_{j-1} +1}^{n}\xi_i(\sigma^\star_{i}) = \\ \nonumber& =\rho_{k_{j-1}^{\star}} +\sum\limits_{i= k^\star_{j-1} +1}^{n}\xi_i(\mu^\star_{j}) =\rho_{k_{j-1}^{\star}} +c_n(\mu^\star_{j};j)
\end{align}
from which using \eqref{1000} it easily follows that the inequality in \eqref{0010} is always satisfied.}


{We are now left with proving that
\begin{equation}\label{last5e1}
 (\sigma_n^{\star}-\sigma_{n+1}^{\star})\Big(\sum\limits_{i=1}^{n}x_i^{\star}-\rho_n\Big)=0 \quad\mathrm{for} \;\; k_{j-1}^{\star}< n \le k_j^{\star}.
\end{equation}
For this purpose, we start observing that \eqref{last5e1} is trivially satisfied for $k_{j-1}^{\star}< n < k_j^{\star}$ due to \eqref{sn_inc_1}.
On the other hand, setting $n=k_j^{\star}$ into \eqref{last5e1} yields
\begin{equation}\label{last5e2}
 (\sigma_{k_j}^{\star}-\sigma_{k_j+1}^{\star})\Big(\sum\limits_{i=1}^{k_j^{\star}}x_i^{\star}-\rho_{k_j^{\star}}\Big)=0
 \end{equation}
which holds true if and only if $\sum\nolimits_{i=1}^{k_j^{\star}}x_i^{\star}-\rho_{k_j^{\star}}=0$ since $\sigma_{k_j}^{\star} - \sigma_{k_j+1}^{\star}>0$. To this end, we observe that
\begin{align}\label{} \nonumber
 \sum\limits_{n=1}^{k_j^{\star}}x_n^{\star}& =\sum\limits_{\ell=0}^{j-1}\sum\limits_{i=k_\ell^{\star}+1}^{k_{\ell+1}^{\star}}x_i^{\star} =\sum\limits_{\ell=0}^{j-1}\sum\limits_{i=k_\ell^{\star}+1}^{k_{\ell+1}^{\star}}\xi_i(\mu_{\ell+1}^{\star}) \\
 & =\rho_{k_{1}^{\star}}+\sum\limits_{\ell=1}^{j-1}\left(\rho_{k_{\ell+1}^{\star}}-\rho_{k_{\ell}^{\star}}\right)=\rho_{k_{j}^{\star}}
\end{align}
which shows that also \eqref{last5e2} is satisfied.

Collecting all the above results together, it follows that the quantities $\{\sigma_{n}^{\star}\}$ computed by means of Algorithm 1 satisfy the KKT conditions.}

\section*{Appendix D \\ Existence and uniqueness of $\varsigma^{\star}_n$}

{The purpose of this Appendix is to show that the quantities $\varsigma^{\star}_{n}$ required by \eqref{101.10} and \eqref{102.10} for the computation of $\mu^\star$ and $k^\star$, respectively, are always well-defined. This amounts to proving that at each iteration either \eqref{100.10} or \eqref{100.101} provide a \emph{unique} $\varsigma^{\star}_{n} \ge 0$ for any $n \in \mathcal{N}_{j}$.}

{As done in Appendix C, we refer to the equivalent form illustrated in Algorithm 2 and start considering the first iteration for which $j=1$, $\gamma_n=\rho_n$ and $k^\star_0 = 0$. Under the assumption that the problem ($\mathcal{P}$) is feasible (see Proposition 1 in Section II.A), and recalling Remark 4 ( see Section III.A), the following two cases are of interest:
{\bf{a}}) $\sum\nolimits_{i=1}^n l_i < \rho_n < \sum\nolimits_{i=1}^n u_i$;
{\bf{b}})  $\rho_n \ge \sum\nolimits_{i=1}^n u_i$.
%
In particular, case {\bf{a}}) can be easily handled observing that $c_n(\varsigma;1)$ is strictly decreasing in the interval $(\omega_n,\Omega_n)$ with
\begin{align}
\omega_n=\min\{h_i(u_i),i=1,\ldots,n\}\\
\Omega_n=\max\{h_i(l_i),i=1,\ldots,n\}
\end{align}
whereas $c_n(\varsigma;1)=\sum\nolimits_{i=1}^n u_i$ for $\varsigma \in [0,\omega_n]$, and $c_n(\varsigma;1)=\sum\nolimits_{i=1}^n l_i$ for $\varsigma \in [\Omega_n,+\infty)$. Accordingly, if case {\bf{a}}) holds true, then the solution of $c_n(\varsigma;1)=\gamma_n$ exists and is unique, it belongs to the interval $(\omega_n,\Omega_n)$, and coincides with the quantity $\varsigma_{n,1}^\star$ as computed through \eqref{100}. On the other hand, if case {\bf{b}}) holds true, then $\varsigma_{n,1}^\star=0$ as given by \eqref{101.0120}. In both cases, Algorithm 2 produces a single value of $\varsigma_{n,1}^\star \ge 0$ for any $n$.

Consider now the $(j+1)$th step of Algorithm 2. Assume that at the $j$th step the value of $\varsigma_{n,j}^{\star}$ is well-defined (in the sense specified above) for any $n > k^\star_{j-1}$. This means that
$\varsigma_{n,j}^{\star}=0$ if
\begin{align}
\gamma_n = \rho_n-\rho_{k^\star_{j-1}} \ge \sum\limits_{i=k^\star_{j-1}+1}^n u_i.
\end{align}
On the other hand, $\varsigma_{n,j}^{\star} >0$ is the unique solution of $c_n(\varsigma;j)=\rho_n-\rho_{k^\star_{j-1}}$, when
\begin{align}
\sum\limits_{i=k^\star_{j-1}+1}^n l_i < \rho_n-\rho_{k^\star_{j-1}} < \sum\limits_{i=k^\star_{j-1}+1}^n u_i.
\end{align}
In the sequel, it is shown that if the above assumptions hold true then the value of $\varsigma_{n,j+1}^{\star}$ for $n > k^\star_{j}$ is also well-defined at the $(j+1)$th step. This amounts to saying that
\begin{align} \label{139}
\sum\limits_{i=k^\star_{j}+1}^n l_i < \rho_n-\rho_{k^\star_{j}}
\end{align}
for $n > k^\star_{j}$.
By contradiction, assume that there exists an index $\bar{k} > k_{j}^\star$ such that 
\begin{align}\label{140}
\sum\limits_{i=k^\star_{j}+1}^{\bar{k}} l_i \ge \rho_{\bar{k}}-\rho_{k^\star_{j}}.
\end{align}
This would mean that $\forall \varsigma >0$
\begin{align}
c_{\bar{k}}(\varsigma; j+1) \ge \sum\limits_{i=k^\star_{j}+1}^{\bar{k}} l_i \ge \rho_{\bar{k}}-\rho_{k^\star_{j}}
\end{align}
from which, recalling that $ c_{\bar{k}}(\varsigma;j+1)=c_{\bar{k}}(\varsigma;j)-c_{k^\star_{j}}(\varsigma;j)$ and setting $\varsigma={\mu}^\star_{j}$,
one would get
 \begin{align}
c_{\bar{k}}({\mu}^\star_{j};j) \ge c_{k^\star_{j}}({\mu}^\star_{j};j) +\rho_{\bar{k}}-\rho_{k^\star_{j}} = \rho_{\bar{k}}-\rho_{k^\star_{j-1}}.
 \end{align}
This would contradict the fact that $c_{n}({\mu}^\star_{j};j) < \rho_{n}-\rho_{k^\star_{j-1}}$ for any $n > {k^\star_{j}}$, as already shown in \eqref{1000}. Accordingly, we must conclude that it cannot exist an index $\bar{k} > k_{j}^\star$ for which \eqref{140} is satisfied, and hence that \eqref{139} holds $\forall n > k_{j}^\star$. In turn, this amounts to saying that if the values of $\varsigma_{n,j}^{\star}$ for $n>k^\star_{j-1}$ are well-defined at the $j$th step, allowing the computation of $\mu^\star_j$ and $k^\star_{j}$, then the values of $\varsigma_{n,j+1}^{\star}$ for $n>k^\star_{j}$, computed at the $(j+1)$th step, are well-defined as well, allowing the computation of $\mu^\star_{j+1}$ and $k^\star_{j+1}$. Since the values of $\varsigma_{n,j}^{\star}$ are well-defined at the first iteration, then they are always well-defined. This concludes the proof.

%
%
%
%
%
%
%

\bibliographystyle{IEEEtran}
\bibliography{IEEEabrv,bibnew}

\end{document}